\def\year{2019}\relax
\documentclass[letterpaper]{article} 
\usepackage{aaai19}  
\usepackage{times}  
\usepackage{helvet}  
\usepackage{courier}  
\usepackage{url}  
\usepackage{graphicx}  
\usepackage{color}
\usepackage{mathrsfs}
\usepackage{amssymb,amsmath,verbatim}
\usepackage{amsfonts,amsthm}
\usepackage {color}
\usepackage{algorithm}
\usepackage{algcompatible}
\usepackage{booktabs,microtype}
\usepackage{tikz}
\usetikzlibrary{arrows,shapes, calc, fit, positioning}
\usepackage{tkz-graph}
\usepackage{etoolbox}
\newcounter{Bew1}
\newcounter{Bew2}

\usepackage{mathtools}
\usepackage{paralist}
\usepackage{xspace}
\usetikzlibrary{backgrounds,automata}
\usetikzlibrary{decorations.pathreplacing}
\usetikzlibrary{calc}
\tikzstyle{simplenode}=[rounded corners=3pt,draw, fill=blue!20, minimum size=2em]
\tikzstyle{taxipath}=[line width=2pt,blue]
\tikzset{mybrace/.style={decorate,decoration={brace,raise=0.5cm}}}
\newcommand{\bbR}{\mathbb{R}}
\newcommand{\bbN}{\mathbb{N}}
\newcommand{\calA}{\mathcal{A}}
\newcommand{\calB}{\mathcal{B}}
\newcommand{\calC}{\mathcal{C}}
\newcommand{\calD}{\mathcal{D}}
\newcommand{\calF}{\mathcal{F}}
\newcommand{\calH}{\mathcal{H}}
\newcommand{\calN}{\mathcal{N}}
\newcommand{\calS}{\mathcal{S}}
\newcommand{\calT}{\mathcal{T}}
\newcommand{\calU}{\mathcal{U}}
\newcommand{\calV}{\mathcal{V}}

\newcommand{\tup}[1]{\langle #1 \rangle}
\newcommand{\eps}{\varepsilon}

\newcommand{\pdim}{\mathit{P}_{\mathit{dim}}}
\newcommand{\poly}{\mathit{poly}}

\newcommand{\desc}{\mbox{\rm desc}}

\newcommand{\height}{\mbox{\rm height}}
\newcommand{\child}{\mbox{\rm child}}
\newcommand{\SAT}{3{\sc SAT}\xspace}
\newcommand{\TBSAT}{{\sc(3,B2)-SAT}\xspace}

\newcommand{\argmax}{\mbox{\rm argmax}}
\newcommand{\blocks}{\mathit{dev}}

\newcommand{\citet}[1]{\citeauthor{#1} (\citeyear{#1})}
\theoremstyle{plain}
	  \newtheorem{theorem}{Theorem}[section]
	  \newtheorem{corollary}[theorem]{Corollary}
	  \newtheorem{lemma}[theorem]{Lemma}

\theoremstyle{definition}
	  \newtheorem{definition}[theorem]{Definition}

\theoremstyle{remark}
	  \newtheorem{remark}[theorem]{Remark}

\begin{document}
%
\title{Forming Probably Stable Communities with Limited Interactions}
\author{%
   Ayumi Igarashi\\
    Kyushu University\\
    Fukuoka, Japan\\
    igarashi@agent.inf.kyushu-u.ac.jp\\
   \And  Jakub Sliwinski\\
   ETH Zurich\\
   Zurich, Switzerland\\
   jsliwinski@ethz.ch\\
   \And Yair Zick\\
   National University of Singapore\\
   Singapore\\
   zick@comp.nus.edu.sg
}

\maketitle
\begin{abstract}
A community needs to be partitioned into disjoint groups; each community member has an underlying preference over the groups that they would want to be a member of. We are interested in finding a stable community structure: one where no subset of members $S$ wants to deviate from the current structure. 
	We model this setting as a hedonic game, where players are connected by an underlying interaction network, and can only consider joining groups that are connected subgraphs of the underlying graph. We analyze the relation between network structure, and one's capability to infer statistically stable (also known as PAC stable) player partitions from data. We show that when the interaction network is a forest, one can efficiently infer PAC stable coalition structures. 
	Furthermore, when the underlying interaction graph is not a forest, efficient PAC stabilizability is no longer achievable. Thus, our results completely characterize when one can leverage the underlying graph structure in order to compute PAC stable outcomes for hedonic games. Finally, given an unknown underlying interaction network, we show that it is NP-hard to decide whether there exists a forest consistent with data samples from the network.
\end{abstract}

\section{Introduction}\label{sec:intro}
A professor wants her students to complete a group programming project. In order to do so, students should divide into project groups with a few students in each; naturally, some groups will be objectively better than others. 
However, students seldom try to find a group that's objectively optimal for them; they would rather join groups that have at least one or two of their friends. This type of scenario falls into the realm of {\em constrained coalition formation}; in other words, how should we partition a group of people given that \begin{inparaenum}[(a)]\item they have preferences over the groups they are assigned to and \item they have limited interactions with one another?\end{inparaenum}~Other scenarios fitting this description include 
\begin{enumerate}[(a)]
	\item Seating arrangements at a wedding (or at conference banquets): some guests should absolutely not be seated together, while others would probably enjoy one another's company. However, it should always be the case that every guest has at least one acquaintance seated at their table. 
	\item Group formation on social media: given a social media network (e.g. Facebook), people prefer being affiliated with certain groups; however, they are limited to joining groups that already contain their friends.
\end{enumerate}
Constrained coalition formation problems are often modeled as {\em hedonic games}. 
Hedonic games formally capture a simple, yet compelling, paradigm: how does one partition players into groups, while factoring individual players' preferences? The literature on hedonic games is primarily focused on finding ``good'' {\em coalition structures} --- partitions of players into disjoint groups. A set of coalition structures satisfying certain desiderata is called a {\em solution concept}. A central hedonic solution concept is {\em coalitional stability}: given a coalition structure $\pi$, we say that a set of players (also known as a {\em coalition}) $S$ can deviate from $\pi$ if every $i\in S$ prefers $S$ to its assigned group under $\pi$; a coalition structure $\pi$ is stable if no coalition $S$ can deviate. In other words, $S$ contains at least one player $i$ who prefers its current coalition (denoted $\pi(i)$) to $S$.  
The set of stable coalition structures --- also known as the {\em core} of the hedonic game --- may be empty; what's worse, even when it is known to be non-empty, finding a stable coalition structure may be computationally intractable. 
Moreover, efficient algorithms for finding stable coalition structures often assume full knowledge of the underlying hedonic game; that is, in order to work, the algorithm needs to have either {\em oracle access to player preferences} (i.e. queries of the form `does player $i$ prefer coalition $S$ to coalition $T$?'), or {\em structural knowledge of the underlying preference structure} (e.g. some concise representation of player preferences that one can leverage in order to obtain a poly-time algorithm). 

Neither assumption is realistic in practice: eliciting user preferences is notoriously difficult, especially over combinatorially complex domains such as subsets of players. If one forgoes preference elicitation and opts for mathematically modeling preferences (e.g. assuming that users have additive preferences over coalition members), it is not entirely obvious what mathematical model of user preferences is valid. 
This leads us to the following natural question: can we find a stable coalition structure when player preferences are unknown? 
Recent works \cite{balcan2015learning,balkanski2017cost,sliwinski2017hedonic} propose a statistical approach to stability in collaborative environments. 
In this framework, one assumes the existence of user preference data over some coalitions, which is then used to construct {\em probably approximately stable} outcomes (the notion is referred to as PAC stability). In this paper, we explore the relation between structural assumptions on player preferences, and computability of PAC stable outcomes. 

{\bf Our contribution} 
We assume that there exists some underlying {\em interaction network} governing player preferences; that is, players are nodes on a graph, and only connected coalitions are feasible. Within this framework, we show that if player preferences are restricted by a forest, one can compute a PAC stable outcome using only a polynomial number of samples. Surprisingly, even if the underlying forest structure is not known to the learner, PAC stabilizability still holds, despite the fact that it may be computationally intractable to find an approximate forest structure that is likely consistent with the true interaction graph. In contrast, we show that it is impossible to find a PAC stable outcome even if the graph contains a single cycle. The latter result is constructive: we show that whenever the underlying interaction graph does contain a cycle, one can construct a sample distribution for which it would be impossible to elicit a PAC stable outcome. 

Our positive result for forests is interesting in several respects. First, while one can find PAC stable outcomes in polynomial time, computing stable outcomes for hedonic games on forests is computationally intractable~\cite{igarashi2016hedonicgraph}; second, unlike \cite{sliwinski2017hedonic}, we do not require that player preferences are provided in the form of numerical utilities over coalitions. This not only makes our results more general, but also more faithful to the problem we model, which assumes ordinal information about player preferences, rather than cardinal utilities.
Finally, in Section~\ref{sec:infer-tree}, we prove a non-trivial technical result on learning forest structures that is of independent interest. Briefly, we study the following problem: we are given samples of subsets of graph vertices, each labeled either `connected' or `disconnected'; we need to decide whether there exists some forest $T^*$ that is consistent with the sample --- i.e. all connected sets of vertices are connected under $T^*$ and all disconnected sets are not. We show that when all of our vertex samples are connected (i.e. we do not observe any disconnected components), it is possible to efficiently learn an underlying forest structure (if one exists); on the other hand, if one assumes that both connected and disconnected sets are presented to the learner, it is computationally intractable to decide whether there exists a forest, or even a path, that is consistent with the samples.

{\bf Related work} 
There exists a rich body of literature studying hedonic games from an economic perspective (e.g. (Banerjee, Konishi, and S\"{o}nmez 2001; Bogomolnaia and Jackson 2002)). More recently, the AI community has be- gun studying both computational and analytical proper- ties of hedonic games (see e.g. (Aziz and Brandl 2012; Deineko and Woeginger 2013; Gairing and Savani 2010; Peters and Elkind 2015), and (Aziz and Savani 2016; Woeginger 2013) for an overview). Interaction networks in cooperative games were first introduced by Myerson (1977). The relation between graph structure and stability in the classic cooperative game setting is also relatively well-understood. Demange (2004) shows that if the underlying interaction network is a forest, then the core is not empty; further studies (Bousquet, Li, and Vetta 2015; Meir et al. 2013) establish relations between approximate stability and the underlying graph structure, while Chalkiadakis, Greco, and Markakis (2016) study the computational complexity of finding core outcomes in graph restricted environments. Igarashi and Elkind (2016) establish both the existence of stable coalition structures in hedonic games over forests, as well as the computational intractability of finding stable coalition structures; Peters (2016) studies the relation between hedonic solution concepts and the treewidth of the underlying interaction graph.

Several works study learning based game-theoretic solu- tion concepts. Sliwinski and Zick (2017) introduce PAC stability in hedonic games, and analyze several common classes of hedonic games. Other works on learning and game theory include learning in cooperative games (Balcan, Procaccia, and Zick 2015; Balkanski, Syed, and Vassilvitskii 2017), rankings (Balcan, Vitercik, and White 2016), auctions (Bal- can et al. 2012; Balcan, Sandholm, and Vitercik 2018; Morgenstern and Roughgarden 2016) and noncooperative games (Fearnley et al. 2013; Sinha, Kar, and Tambe 2016).

\section{Preliminaries}\label{sec:prelim}
Throughout this paper, vectors are denoted by $\vec x$, and sets are denoted by uppercase letters; given a value $s \in \bbN$, we set $[s]=\{1,\dots,s\}$. A {\em hedonic game} is given by a pair $\tup{N, \succeq}$, where $N = [n]$ is a finite set of {\em players}, and $\succeq = (\succeq_1,\dots,\succeq_n)$ is a list of preferences players in $N$ have over subsets of $N$ (also referred to as {\em coalitions}); in more detail, for every $i \in N$, we write $\calN_i=\{\, S \subseteq N \mid i \in S \,\}$; $\succeq_{i}$ describes a complete and transitive preference relation over $\cal N_i$. 
For each $i \in N$, let $\succ_{i}$ denote the strict preference derived from $\succeq_{i}$, i.e., 
$S \succ_i S^{\prime}$ if $S \succeq_i S^{\prime}$, but $S^{\prime} \not \succeq_{i} S$. 
An outcome of a hedonic game is a {\em coalition structure}, i.e., a partition $\pi$ of $N$ into disjoint coalitions; we denote by $\pi(i)$ the coalition containing $i \in N$. A solution concept is a mapping whose input is a hedonic game $\tup{N,\succeq}$, and whose output is a (possibly empty) set of coalition structures. 
The {\em core} is the most fundamental solution concept in hedonic games. First, we say that a coalition $S$ {\em strongly blocks} a coalition structure $\pi$ if every player $i \in S$ strictly prefers $S$ to its current coalition $\pi(i)$, i.e. $S \succ_i \pi(i)$. A coalition structure $\pi$ is said to be {\em core stable} if no coalition $S \subseteq N$ strongly blocks $\pi$. 
\subsection{Interaction Networks}\label{sec:interactionnetworks}
Given an undirected graph $G = \tup{N,E}$ whose nodes are the player set, we restrict the space of {\em feasible coalitions} to be the set of connected subsets of $G$; we denote by $\calF_E$ the set of feasible coalitions. 
Intuitively, we restrict our attention to coalition structures where all group members form a social subnetwork of the underlying interaction graph. Note that when $G$ is a clique, all coalitions are feasible, and the result is a standard (unrestricted) hedonic game. From now on, we define a {\em hedonic graph game} as the tuple $\tup{N,\succeq,E}$; here, $N$ is the set of players, $\succeq$ their preference relations, and $E$ the edges of the underlying interaction network. We focus our attention only on core stable coalition structures that consist of feasible coalitions.

In what follows, it is useful to express player preferences in terms of cardinal utilities. In other words, player $i$ assigns a value $v_i(S) \in \bbR$ to every coalition $S \in \calN_i$; we write a hedonic game as $\tup{N,\calV}$ where $\calV$ is a collection of functions $v_i: \calN_i \rightarrow \bbR$ for each $i \in N$. This representation allows us to seamlessly integrate ideas from PAC learning into the hedonic games model, and is indeed quite common in other works studying hedonic games. However, as we later show, our main result (Theorem~\ref{thm:hgtrees-stabilizable}) still holds when we transition from a utility-based cardinal model, to a preference-based ordinal model. 

\subsection{PAC Learning}\label{sec:PAClearning}
We provide a brief introduction to PAC learning\footnote{What we show here is but one of many variants on the theory of PAC learning. There are many excellent sources on this classic theory; we refer our reader to~\cite{anthony1999learning,kearns1994introtoclt,shashua2009introtoml}}. The basic idea is as follows: we are given an unknown function $v:2^N \rightarrow \bbR$ (a {\em target concept} in the language of PAC learning) that assigns values to subsets of players. In addition, we are given a set of $m$ samples $((S_1,v(S_1)),\dots,(S_m,v(S_m))$ where $S_j \subseteq N$ and $v(S_j)$ is the valuation of $v$ over $S_j$; we wish to estimate $v$ on subsets we did not observe. We assume that $v$ belongs to a {\em hypothesis class} $\calH$ (say, we know that $v$ is an additive valuation). 
Our goal is to output a {\em hypothesis} $v^* \in \calH$ (e.g. if $v$ is additive, $v^*$ should be as well) that is likely to match the outputs of $v$ on future observations drawn from some distribution $\calD$. More formally, a hypothesis $v^*$ is {\em $\epsilon$ approximately correct} w.r.t a probability distribution $\calD$ over $2^N$ and an unknown function $v$ if 
$$
\Pr_{S\sim \calD}[v^*(S) \ne v(S)] < \epsilon.
$$
A learning algorithm $\calA$ takes as input $m$ samples 
\[
(S_1,v(S_1)), (S_2,v(S_2)),\dots,(S_m,v(S_m))
\]
drawn i.i.d. from a distribution $\calD$ over $2^N$, and two parameters $\epsilon,\delta > 0$. 

A class of functions $\calH$ is $(\epsilon,\delta)$ {\em PAC (probably approximately correctly) learnable} if there exists an algorithm $\calA$ that for any $v \in \calH$ and probability distribution $\calD$ over $2^N$, with probability of at least $1- \delta$, it outputs a hypothesis $v^*$ that is $\epsilon$ approximately correct with respect to $\calD$ and $v$. If this holds for any $\epsilon,\delta>0$, $\calH$ is said to be {\em PAC learnable}; moreover, if the running time of $\calA$, and the number of samples $m$ are polynomial in $\frac1 \epsilon,\log\frac1\delta$ and $n$, $\calH$ is said to be {\em efficiently PAC learnable}.

The value $\delta$ is the {\em confidence parameter}: intuitively, it is the probability that the random samples drawn from $\calD$ do not accurately portray the true sample distribution; for example, if $\calD$ is the uniform distribution, then it is possible (though unlikely) that we draw the same subset in every one of our $m$ samples. The value $\epsilon$ is called the {\em error parameter}: it is the likelihood that our hypothesis $v^*$ does not agree with the target concept $v$. Not all hypothesis classes are efficiently PAC learnable; learnability is inherently related to the complexity of the hypothesis class. The complexity of real-valued functions is commonly measured using the notion of {\em pseudo dimension}~(see e.g. Chapter $11$ of \cite{anthony1999learning}).
Given a list of sets $S_1,\dots, S_m\subseteq N$, and corresponding values $r_1,\dots,r_m\in \bbR$ we say that a class of functions $\calH$ can {\em pseudo-shatter} $\left({S_j,r_j}\right)_{j=1}^m$ if for any labeling $\ell_1,\dots,\ell_m \in \{0,1\}$, there is some $v \in \calH$ such that $v(S_j) \ge r_j$ iff $\ell_j = 1$. The {\em pseudo-dimension} of $\calH$, denoted $\pdim(\calH)$ is 
\[
\max\{\, m \mid \exists \left({S_j,r_j}\right)_{j=1}^m \mbox{ that can be shattered by }\calH\,\}.
\]

The following well-known theorem relates the pseudo-dimension and PAC learnability.

\begin{theorem}[\cite{anthony1999learning}]\label{thm:pdim}
	A class of functions $\calH$ is efficiently $(\epsilon,\delta)$ PAC learnable using $m=\poly(\pdim(\calH),\frac 1 \epsilon,\log \frac 1 \delta)$ samples if there exists an algorithm such that given $m$ samples $(S_1,v(S_1)), (S_2,v(S_2)),\dots,(S_m,v(S_m))$ drawn i.i.d. from a distribution $\calD$, it outputs $v^*\in \calH$ consistent with the sample, i.e. $v^*(S_j) = v(S_j)$ for all sampled $S_j$, and runs in time polynomial in $\frac1 \epsilon,\log\frac1\delta$ and $m$. Furthermore, if $\pdim(\calH)$ is superpolynomial in $n$, $\calH$ is not PAC learnable. 
\end{theorem}
In other words, in order to establish the PAC learnability of some hypothesis class, it suffices that one shows that its pseudo dimension is low, and that there exists some efficient algorithm that is able to output a hypothesis $v^*$ which matches the outputs of $v$ on all samples. We note that even if an efficient consistent algorithm does not exist (e.g. if the problem of matching a hypothesis to the samples is computationally intractable), a low pseudo dimension is still desirable: it implies that the number of samples needed in order to find a good hypothesis is polynomial. 

\subsection{PAC Stabilizability}\label{sec:pacstable}
When studying hedonic games, one is not necessarily interested in eliciting approximately accurate user preferences over coalitions using data; in our case, we are interested in identifying core stable coalition structures. Intuitively, it seems that the following idea might work: first, infer player utilities from data and obtain a PAC approximation of the original hedonic game; next, find a coalition structure that stabilizes the approximate hedonic game. This approach, however, may be overcomplicated: first, it may be impossible to PAC learn player preferences from data (this depends on the hypothesis class); moreover, computing a core coalition structure for the learned game may be computationally intractable.  
\cite{sliwinski2017hedonic} propose learning a stable outcome directly from data. They introduce a statistical notion of core stability for hedonic games, which they term {\em PAC stability} (this term was first used by \cite{balcan2015learning} for cooperative transferable utility games).

We say that a partition $\pi$ is $\epsilon$-PAC stable w.r.t. a probability distribution $\calD$ over $2^N$ if 
$$\underset{S \sim \calD}{\Pr}[S \text{ strongly blocks } \pi] < \epsilon.$$ 
The inputs to our learning algorithms will be samples $$(S_1,\vec v(S_1)),(S_2,\vec v(S_2)),\dots,(S_m,\vec v(S_m)),$$ 
where $S_1,\dots,S_m\subseteq N$, and $\vec v(S_j)$ is a vector describing players' utilities over $S_j$; that is, $\vec v(S_j)= (v_i(S_j))_{i \in S_j}$. 

Given an unknown hedonic game $\tup{N,\calV}$ belonging to some hypothesis class $\calH$, a {\em PAC stabilizing algorithm} $\calA$
takes as input $m$ sets $S_1,\dots, S_m$ sampled i.i.d. from a distribution $\calD$, and players' preferences over the sampled sets; in addition, it receives two parameters $\epsilon,\delta > 0$. The algorithm $\calA$ {\em PAC stabilizes} $\calH$, if for any hedonic game $\tup{N,\calV} \in \calH$, distribution $\calD$ over $2^N$, and parameters $\epsilon,\delta > 0$, with probability $\ge 1-\delta$, $\calA$ outputs an $\epsilon$-PAC stable coalition structure if it exists; again, if the running time of the algorithm $\calA$ and the number of samples, $m$, are bounded by a polynomial in $n$, $\frac{1}{\epsilon}$ and $\log \frac{1}{\delta}$, then we say that $\calA$ {\em efficiently PAC stabilizes} $\calH$. 
Similarly, we say that $\calH$ is {\em (efficiently) PAC stabilizable} if there is some algorithm $\calA$ that (efficiently) PAC stabilizes $\calH$.

\section{Learning Hedonic Graph Games}\label{sec:learngraph}
In what follows we consider the following hypothesis class. 
\begin{definition}\label{def:HedonicG}
	For an undirected graph $G = \tup{N,E}$, let $\calH_G$ be the class of all hedonic games $\tup{N,\calV}$ where for each player $i \in N$, $v_i(\{i\})=0$ and player $i$ strictly prefers its singleton to any disconnected coalition $S  \in \calN_i \setminus \calF_E$, i.e., $v_i(S) < 0$ for all $S  \in \calN_i \setminus \calF_E$.
\end{definition}
We first present a baseline negative result: fixing a forest $G$, the hypothesis class, $\calH_G$ is not efficiently PAC learnable. When referring to the PAC learnability of any class of hedonic games, we mean inferring some utility function $v_i^*:2^N \to \bbR$ for all $i \in N$ that PAC approximates the true utilities of players in $N$. This approximation guarantee can be interpreted in both an ordinal and cardinal manner. If we are given player $i$'s ordinal preferences, this simply means that $v_i^*$ is consistent with the ordinal preferences; if we are given player $i$'s cardinal utility function $v_i$, $v_i^*$ should be a PAC approximation of $v_i$. As Theorem~\ref{thm:graph-unlearnable} shows, even when we are given additional information about the underlying graph interaction network, players' preferences are not PAC learnable.
\begin{theorem}\label{thm:graph-unlearnable}
	For any graph $G=\tup{N,E}$ with exponentially many connected coalitions, the class $\calH_G$ is not efficiently PAC learnable.
\end{theorem}
\begin{proof}
Recall that $\calF_E$ is the set of all feasible coalitions over $G=\tup{N,E}$; by assumption, $|\calF_E|$ is exponential. Let $\calH_i$ be the set of all possible utility functions $v_i:\calN_i \rightarrow \bbR$ satisfying $v_i(\{i\})=0$ and $v(S) < 0$ for all disconnected coalition $S \in \calN_i \setminus \calF_E$. The utility player $i$ derives from feasible coalitions in $G$ is unrestricted; in particular, one cannot deduce anything about the utility of some feasible coalition $S\in \calF_E$, based on other feasible coalitions' utilities. This immediately implies that the set $\calF_E$ can be pseudo-shattered by $\calH_i$. Hence $\pdim(\calH_i)$ is at least exponential, and by Theorem~\ref{thm:pdim}, $\calH_i$ is not efficiently PAC learnable.
\end{proof}

As an immediate corollary, forest interaction structures do not admit PAC learnable preference structures in general; this is true even if $G$ is a star graph over $n$ players, since the number of feasible coalitions is exponential in $n$.
\begin{corollary}\label{cor:star}
	Let $G$ be a star graph over $n$ players; then $\calH_G$ is not PAC learnable.
\end{corollary}
\begin{proof}
For a star with $n$ nodes, any coalition containing the center of the star is feasible, hence it has $2^{n-1}$ feasible coalitions. By Theorem \ref{thm:graph-unlearnable}, hedonic games on forests are not PAC learnable.
\end{proof}

The reason that hedonic games with forest interaction structures are not PAC learnable is that they may have exponentially many feasible coalitions; this is also the reason that finding a core stable coalition structure for hedonic games with forest interaction structures is computationally intractable~\cite{igarashi2016hedonicgraph}. However, we now show how one can still exploit the structural properties of forest graph structures to efficiently compute PAC stable outcomes.

\section{PAC Stabilizability of Hedonic Graph Games}\label{sec:hg-stable}
Having established that hedonic games with a forest interaction structure are not, generally speaking, PAC learnable, we turn our attention to their PAC stabilizability. We divide our analysis into two parts. We begin by assuming that the underlying interaction graph structure $G$ is known to us; in other words, we know that our game belongs to the hypothesis class $\calH_G$. In Section~\ref{sec:infer-tree}, we show how one can forgo this assumption. 

\begin{theorem}\label{thm:hgtrees-stabilizable}
	If $G = \tup{N,E}$ is a forest, $\calH_G$ is efficiently PAC stabilizable.
\end{theorem}
\begin{proof}
We claim that Algorithm~\ref{alg:tree} PAC stabilizes $\calH_G$. It is related to the algorithm introduced in \citet{demange2004stability} used to find core stable outcomes for forest-restricted hedonic games\footnote{\citet{demange2004stability} presents the algorithm for non-transferable cooperative utility games on trees where each coalition has a choice of action. A hedonic game is a special case of a non-transferable utility game where each coalition has a unique action.} in the full information setting. Intuitively, instead of identifying the {\it guaranteed coalition} for each player precisely, Algorithm~\ref{alg:tree} approximates it. If the input graph $\tup{N, E}$ is a forest, we can process each of its connected components separately, so we can assume that $\tup{N, E}$ is a tree.
	
We first provide an informal description of our algorithm, followed by pseudocode. The algorithm first transforms $\tup{N,E}$ into a rooted tree with root $r$ by orienting the edges in $E$ towards the leaves. For every player $i$ starting from the bottom to the top, the algorithm identifies $B_i$ - a coalition containing $i$, the best for $i$ observed in the samples that is entirely contained in $i$'s subtree, such that others in $B_i$ prefer it to their own best guaranteed coalition; in other words, $B_i \succeq_j B_j$ for all $j \in B_i$. Having identified $B_i$ for every $i \in N$, players are partitioned according to the $B_i$'s from top-down. The main concern is to ensure that $B_i$ is a good approximation of its full-information counterpart; this is guaranteed by taking a sufficiently large sample size $m=\lceil\frac{n}{\epsilon}\log\frac{n}{\delta}\rceil$. 
	
In what follows, we assume an orientation of the trees in $G$, with arbitrary root nodes. Fixing the orientation, we let $\desc(i)$ be the set of descendants of $i$ (we assume that $i \in \desc(i)$). 
For each coalition $S \subseteq N$, we denote by $\child(S)$ the set of {\em children} of $S$, namely, 
\[
\child(S)=\{\, i \in N \setminus S \mid~\mbox{$i$'s parent belongs to}~S \,\}.
\]
The {\it height} of a node $i \in N$ is defined inductively as follows: $\height(i):=0$ if $i$ is a leaf, i.e., $\desc(i)=\{i\}$, and 
\[
\height(i):=1+\max \{\, \height(j)\mid j \in \desc(i) \setminus  \{i\} \,\}, 
\]
otherwise.
	
\begin{algorithm}[h]
	\begin{algorithmic}[1]
		\STATEx \textbf{Input:} set $\calS$ of $m=\lceil\frac{n}{\epsilon}\log\frac{n}{\delta}\rceil$ samples from $\cal D$
		\STATEx \textbf{Output:} a partition $\pi=\pi^{(r)}$ of $N$
		\STATE Make a rooted tree with root $r$ by orienting all the edges in $E$ towards the leaves.
		\STATE Initialize $B_i \gets \emptyset$ and ${\pi}^{(i)} \gets \emptyset$ for each $i \in N$.
		\FOR{$t=0,\ldots,\height(r)$}
		\FOR{each node $i \in N$ with $\height(i)=t$}
		\STATE set $\calS^*=\{\, S \in \calS \cap \calF_E \mid i \in S \subseteq \desc(i)~\land~v_j(S) \ge v_j(B_j),~\mbox{for all}~ j \in S \setminus \{i\} \,\}\cup\{\{i\}\}$ \label{step:guarantee} 
		\STATE choose $B_i \in \argmax \{\, v_i(S) \mid S \in \calS^* \cup \{\{i\}\}\,\}$ \label{step:best}
		\STATE set $\pi^{(i)} \gets \{B_i\} \cup \bigcup \{\, {\pi}^{(j)} \mid j \in \child(B_i) \,\}$
		\ENDFOR
		\ENDFOR
	\end{algorithmic}
	\caption{An algorithm finding a PAC stable outcome for forest-restricted games}\label{alg:tree}
\end{algorithm}

Given player $i$, let $\calB_i$ be the collection of coalitions $B_j$ for every descendant $j \neq i$ of $i$, i.e., $\calB_i=\{\, B_j \mid j \in \desc(i)\setminus \{i\}\,\}$. 
	For each $i \in N$ and each coalition $X \subseteq N$, we let $P_{\calB_i}(X)$ mean that $i \in X\subseteq \desc(i)$, $X$ is connected, and every other player $j$ in $X\setminus \{i\}$ weakly prefers $X$ to $B_j$. Now, we define a modified preference order for player $i$, $\succeq_{\calB_i}$, that devalues any coalition $X$ for which $P_{\calB_i}(X)$ does not hold.
	\begin{itemize}
		\item If $P_{\calB_i}(X)$ and $P_{\calB_i}(Y)$, then $X \succ_{\calB_i} Y \iff X \succ_i Y$
		\item If $P_{\calB_i}(X)$ but $\neg P_{\calB_i}(Y)$, then $X \succ_{\calB_i} Y$
		\item If $\neg P_{\calB_i}(X)$, then $\forall Y: Y \succeq_{\calB_i} X$
	\end{itemize}

	Given $\calB_i$ and a distribution $\calD$, we say that a coalition $X$ is {\em top-$\frac{\epsilon}{n}$} for player $i$, if 
	\[
	\underset{S \sim \calD}{\Pr}[S \succ_{\calB_i} X] \le \frac{\epsilon}{n}.
	\]
	Trivially, for every $\calB_i$ the probability of sampling a top-$\frac{\epsilon}{n}$ coalition for player $i$ from $\calD$ is at least $\frac{\epsilon}{n}$; moreover, if $\underset{S \sim \calD}{\Pr}[P_{\calB_i}(S)] \le \frac{\epsilon}{n}$, then any coalition is top-$\frac{\epsilon}{n}$.
	
	Intuitively, $B_i$ approximates the best coalition $i$ can form with members of the subtree rooted at $i$. Algorithm~\ref{alg:tree}'s objective is to ensure that sampling a coalition $S$ from $\calD$ such that $P_{\calB_i}(S) \land S \succ_i B_i$ is unlikely, namely, the probability of seeing $S$ from $\calD$ such that $S$ is better for the highest node $i$ in $S$ than $B_i$, and every other player in $S$ prefers it to their $B_j$, is smaller than $\epsilon$; this is done by examining enough coalitions so as to see some top-$\frac{\epsilon}{n}$ coalition for every player.
	
	Examine what happens if $\calB_i$ containing $B_j$'s for $i$'s descendants is fixed upfront, i.e. not dependent on the sample. Let us bound the probability that for $i$, none of the coalitions in $\calS$ are top-$\frac{\epsilon}{n}$:
	\begin{align}\label{eq1}
	&\left(1-\frac{\epsilon}{n}\right)^m = \left( 1-\frac{\epsilon}{n}\right)^{\lceil\frac{n}{\epsilon}\log\frac{n}{\delta} \rceil } \\
	&\le \left( \left( 1-\frac{\epsilon}{n}\right)^{\frac{n}{\epsilon}} \right) ^{\log\frac{n}{\delta}}< \left( \frac{1}{e} \right)^{\log\frac{n}{\delta}} < \frac{\delta}{n} \nonumber
	\end{align}
	Note that Inequality~\eqref{eq1} is true irrespective of what $\calB_i$ is. Taking a union bound, the probability that there is some player $i$ such that there is no top-$\frac{\epsilon}{n}$ coalition for $i$ in $\calS$ is at most $\delta$. Note that $\calS^*$ can end up not containing any coalition (line \ref{step:guarantee}). But then with high confidence, as a special case of the above consideration, every coalition is top-$\frac{\epsilon}{n}$, and the algorithm can pick $\{i\}$.
	
	Recall that in an actual run of the algorithm the sample $\calS$ is drawn, and for every descendant $j$ of $i$, $B_j$ is computed based on $\calS$, and then $B_i$ is computed based on the same sample. One can ask whether some dependence between the computation of $B_i$ and the $B_j$'s does not invalidate Inequality~\eqref{eq1}. This potential problem can be easily solved by taking a larger number of samples: if we take $m = \lceil\frac{n^2}{\epsilon}\log \frac n\delta \rceil$ samples, we can just use $\frac{n}{\epsilon}\log\frac{n}{\delta}$ samples to compute each $B_i$ and maintain complete independence in the samples. 
	
	In order to see the smaller sample size used in Algorithm~\ref{alg:tree} provides the same guarantee, consider an equivalent reordering of the computation of $B_i$ and $B_j$'s: first, for every $i \in N$, determine the number $k$ of connected coalitions $S$ in the sample such that $i$ will be the highest node in $S$. Then, draw the other $m-k$ coalitions and compute $B_j$'s for every descendant $j$ of $i$; finally, based on this, determine the family $\calB_i$. Note that regardless of what $\calB_i$ is, each of the undetermined, independently drawn $k$ coalitions has probability of at least $\frac{\epsilon}{n}$ to be top-$\frac{\epsilon}{n}$ for $i$. Hence, the inequality \eqref{eq1} holds even if $\calB_i$ and $B_i$ are computed based on the same sample of coalitions $\calS$.
	
	We are now ready to prove that the coalition structure outputted by Algorithm~\ref{alg:tree} returns a PAC stable outcome $\pi^{(r)}$. We observe that any coalition included in the returned $\pi^{(r)}$ is a $B_i$ for some $i$. 
	Note that for every $j \in B_i$, we have that $v_j(B_i) \ge v_j(B_j)$ (line \ref{step:guarantee}). 
	Now, consider any coalition $X$ that strongly blocks $\pi^{(r)}$; let $i = \argmax_{j \in X} \height(j)$. 
	Since $X$ strongly blocks $\pi^{(r)}$, 
	\[
	v_j(X) > v_j(\pi^{(r)}(j)) \ge v_j(B_j)
	\]
	for all players $j \in X$.
	In particular, $v_i(X) > v_i(B_i)$. By construction of $B_i$ and Inequality~\eqref{eq1}, $B_i$ is top-$\frac{\epsilon}{n}$ for $i$; that is, $$\frac{\epsilon}{n}> \Pr_{S\sim \calD}[v_i(S) > v_i(B_i)] \ge \Pr_{S\sim \calD}[S=X];$$ thus the probability of drawing a coalition such as $X$ from $\calD$, i.e. strongly blocking $\pi^{(r)}$ and having $i = \argmax_{j \in X} \height(j)$, is less than $\frac{\epsilon}{n}$. Taking a union bound over all players, $$\underset{X \sim \calD}{\Pr}[X \text{ strongly blocks }\pi^{(r)}] < \epsilon;$$ this guarantee holds with confidence $1 - \delta$.
\end{proof}

We conjecture that a similar argument can imply a stronger statement. That is, we can replace `strongly block' in the definition of PAC stabilizability with `weakly block' and still obtain PAC stabilizability on trees. (A coalition $S$ weakly blocks a coalition structure if every player weakly prefers $S$ to their current coalition and at least one player in $S$ has a strict preference) We note that in the full information setting, a strict core outcome does not necessarily exist on trees \cite{igarashi2016hedonicgraph}.
\begin{remark}[From Cardinal to Ordinal Preferences]
	Note that step \ref{step:best} of the Algorithm \ref{alg:tree} is the {\em only} step that refers to the numerical representation of agent preferences $v_i$. The algorithm chooses a coalition with maximal utility value $v_i$ out of some set of possible coalitions; in particular, the only thing required for the successful implementation of Algorithm~\ref{alg:tree} is players' ranking of coalitions in the sample. In other words, the particular numerical representation of player preferences plays no role. This is a significant departure from the algorithms devised by \cite{sliwinski2017hedonic}, where the type of utility representation functions used was crucial for PAC stability.
\end{remark}
Next, we show that Theorem~\ref{thm:hgtrees-stabilizable} is `tight' in the sense that if the graph $G$ contains a cycle, $\calH_G$ is not PAC stabilizable. 
\begin{theorem}\label{thm:notforest-notstable}
	Given a non-forest graph $G = \tup{N,E}$, the class $\calH_G$ is not PAC stabilizable.
\end{theorem}
\begin{proof}
Since $G$ is not a forest, there is a cycle in $G$. Without of loss of generality, let $C=\{1,2,\ldots,k\}$ be a cycle with $\{i,i+1\} \in E$ for all $i=1,2,\ldots,k-1$, and $\{k,1\} \in E$. Let $S_1 = \{1, 2\}$, $S_2 = \{2, 3\}$, $S_3 = \{3,\ldots, k,1\}$. Suppose $\calD$ is the uniform distribution on $\{S_1, S_2, S_3\}$ and that the following holds:
	\begin{align} 
	S_1 \succ_{1} S_3, S_2 \succ_{2} S_1, S_3 \succ_{3} S_2.\label{eq:hg-unstab}
	\end{align}
	In this case, nothing beyond \eqref{eq:hg-unstab} can be deduced about the game by examining samples from $\calD$. Consider the following games satisfying \eqref{eq:hg-unstab}:
	\begin{itemize}
		\item A game $\Gamma_1$ where every player $i\in \{1,2,3\}$ strictly prefers $\{i\}$ to any other coalition, and any non-singleton coalition is less preferred than $S_i$ and $S_{i-1}$, namely, $\{i\}\succ_{i}  S_i \succ_{i} S_{i-1} \succ_{i} S^{\prime}$ for any $S^{\prime} \in \calN_i \setminus \{S_i,S_{i-1},\{i\}\}$. Here we set $S_0=S_3$. 
		Every player $j \in S_3$ strictly prefers $S_3$ to any other coalition.
		\item A game $\Gamma_2$ where every player in $C$ strictly prefers $C =\{1, 2, ... , k\}$ to any other coalition, and every player $i\in \{1,2,3\}$ strictly prefers $S_i$ to any coalition other than $C$. Every player $j \in S_3$ strictly prefers $S_3$ to any other coalition other than $C$.
	\end{itemize}
	Suppose towards a contradiction that there is an algorithm $\calA$ that returns a $\frac{1}{3}$-PAC stable partition $\pi$. We will show that for $\pi$ to be resistant against deviations supported by $\cal D$, $\pi$ has to include $\{1\}$ or $\{2\}$ or $\{3\}$ for the first game, and $C$ for the second game, which implies that it is impossible to achieve $\eps <\frac{1}{3}$ with any confidence $1-\delta > 0$.
	
	\begin{itemize}
	\item Consider the first game $\Gamma_1$. Suppose for a contradiction that no player $i \in \{1, 2, 3\}$ forms a singleton. We will show that at least one of $S_1$, $S_2$, and $S_3$ would strongly block $\pi$ with probability $1$. The claim is clear when no player $i \in \{1, 2, 3\}$ belongs to $S_i$; thus suppose at least one of $S_i$ is formed. Then we have the following three cases.
\begin{itemize}
\item If $\pi(1) = S_1$, $\pi(2) = S_1$, and $\pi(3) \neq \{3\}$, then players in $S_2$ strictly prefer $S_2$ to their own coalitions.
\item If $\pi(2) = S_2$, $\pi(3) = S_2$, and $\pi(1) \neq \{1\}$, then players in $S_3$ strictly prefer $S_3$ to their own coalitions.
\item If $\pi(j)=S_3$ for all $j \in S_3$, and $\pi(2?=\{2\}$, then players in $S_1$ strictly prefer $S_1$ to their own coalitions.
\end{itemize}
In either case, $\pi$ is strongly blocked with probability at least $\frac{1}{3}$, a contradiction.
	\item Consider the second game $\Gamma_2$. Suppose for a contradiction that the coalition $C$ is not formed, i.e., $C \not \in \pi$. Again, at least one of $S_i$ is formed as otherwise $\pi$ would not be resistant against deviations supported by $\calD$. Now we have the following three cases.
	\begin{itemize}
\item If $\pi(1) = S_1$, $\pi(2) = S_1$, then players in $S_2$ strictly prefer $S_2$ to their own coalitions.
\item If $\pi(2) = S_2$, $\pi(3) = S_2$, then players in $S_3$ strictly prefer $S_3$ to their own coalitions.
\item If $\pi(j)=S_3$ for all $j \in S_3$, then players in $S_1$ strictly prefer $S_1$ to their own coalitions.
\end{itemize}
In either case, $\pi$ is strongly blocked with probability at least $\frac{1}{3}$, a contradiction.
	\end{itemize}
\end{proof}

\section{Inferring Tree Interaction Networks from Data}\label{sec:infer-tree}
Until now, we assume that the underlying interaction network $G$ was given to us as input; this is, naturally, an assumption that we would like to forgo. Suppose the underlying graph is a forest $T=\tup{N,E}$, and consider the question of whether it is possible to infer a forest $T^* = \tup{N,E^*}$ that agrees with the original graph with high probability. Let $\calT_n$ be the set of all possible trees over $n$ vertices, and let $\calF_n$ be the set of all possible forests; $\calF_n$ is our hypothesis class for guessing an approximate forest. More formally, $\calF_n$ consists of functions $f_G$ that given an $n$ vertex forest $G$, output $1$ if a set of vertices is connected, and 0 otherwise.
By Cayley's formula:
\begin{align}
|\calT_n| = n^{n-2}
\end{align}

Any forest can be obtained by choosing a tree, and then choosing a subset of its edges, hence:

\begin{align}
|\calF_n| \le |\calT_n| 2^{n-1} = n^{n-2} 2^{n-1} \label{lem:num-of-forests}
\end{align}

We observe the following variant of Theorem~\ref{thm:pdim} for finite hypothesis classes. 
\begin{theorem}[\citet{anthony1999learning}]\label{thm:finite-learn}
	Let $\calC$ be a finite hypothesis class where $\log|\calC|$ is polynomial in $n$. If there exists a polynomial time algorithm that for any $v \in \calC$, and samples $$\tup{S_1,v(S_1)}\dots,\tup{S_m,v(S_m)}$$ 
	finds a function $v^*\in \calC$ consistent with the samples, i.e., $v^*(S_j)=v(S_j)$ for each $j=1,2,\ldots,m$, then $\calC$ is efficiently PAC learnable.
\end{theorem}
Since $\log|\calF_n| < 2n \log n$, all we need is to establish the existence of an efficient algorithm to compute a forest consistent with a given sample. More formally, let $T$ be an unknown forest; we are given a set of $m$ subsets of vertices labeled 'connected' or 'disconnected' according to $T$, can we find a forest that is consistent with the labeling? First, we consider an easier question and assume all subsets are connected. The answer to this question is affirmative, and appears in \citet{conitzer2004graphs}.

\begin{theorem}[\citet{conitzer2004graphs}]\label{thm:conitzergraphs}
	Let $T = \tup{N,E}$ be a tree. Given a list $S_1,\dots,S_m$ of connected vertices in $T$, there exists a poly-time algorithm that outputs a tree $T^*$ where every subset $S_j$ is connected in $T^*$. 
\end{theorem}

Theorem~\ref{thm:conitzergraphs} pertains to trees, but immediately generalizes to forests by noting that if $T$ is a forest, any tree whose edgeset is a superset of $E$ is a valid solution as well, hence the same algorithm solves the problem.

In other words, if one only observes subsets of feasible coalitions and players' preferences over them, it is possible to find a forest structure consistent with the samples.

\begin{corollary}\label{cor:forests-learnable}
	If the probability distribution $\calD$ supports only connected subgraphs, $\calF_n$ is efficiently PAC learnable over $\calD$.
\end{corollary}

Corollary~\ref{cor:forests-learnable} is immediately implied by \eqref{lem:num-of-forests}, Theorems \ref{thm:finite-learn} and \ref{thm:conitzergraphs}.
Theorem~\ref{thm:hgtrees-stabilizable} assumes that the underlying interaction graph is known to us. Leveraging Corollary~\ref{cor:forests-learnable}, we now show that this assumption can be forgone; that is, it is possible to PAC stabilize a hedonic game whose underlying interaction graph is a forest, even if the forest structure is unknown to us. Note that we established that the forest structure can be PAC learned efficiently only if the sample contains exclusively connected coalitions, yet we do not have this requirement for PAC stabilizability.

\begin{theorem}\label{thm:hg-stabilizable}
	Let $\calH^* = \bigcup\{\, \calH_G \mid \mbox{$G$ is a forest} \,\}$ be the class of all hedonic games whose interaction graph is a forest; then $\calH^*$ is efficiently PAC stabilizable.
\end{theorem}
\begin{proof}
	Suppose $\eps, \delta$ are given, and there is an unknown forest $G$, hedonic game $\tup{N,\calV} \in \calH_G$ and a probability distribution $\cal D$ over coalitions. Let $\cal D'$ be a distribution obtained from $\cal D$ by substituting any disconnected coalitions with $\emptyset$. $\cal D'$ supports only connected coalitions, so by Corollary \ref{cor:forests-learnable}, $G$ can be efficiently PAC learned with respect to $\cal D'$ to obtain $G'$ s.t. $\Pr_{S \sim \cal D'}[f_{G'}(S) \ne f_G(S)] < \frac{\epsilon}{2}$ with confidence $1-\frac{\delta}{2}$. Let $\cal D''$ be a distribution obtained from $\cal D'$ by substituting any coalitions $S$ s.t. $f_{G'}(S) \ne f_G(S)$ with $\emptyset$. Since $f_{G'}(S) = f_G(S)$ for any $S$ supported by $\cal D''$, by Theorem \ref{thm:hgtrees-stabilizable}, given $G'$, $\tup{N,\calV}$ can be PAC stabilized with respect to $\cal D''$ to obtain a partitioning of the agents $\pi$ such that $\Pr_{S \sim \cal D''}[S \text{ strongly blocks } \pi] < \frac{\epsilon}{2}$ with confidence $1-\frac{\delta}{2}$. For ease of notation, we write $\blocks(S,\pi)$ whenever $S$ strongly blocks $\pi$. 
	\begin{align}
	\Pr_{S \sim \cal D}[\blocks(S,\pi)] \notag  =&  \Pr_{S \sim \cal D}[\blocks(S,\pi)\land \mbox{$S$ is connected in $G$}]\\
	=& \Pr_{S \sim \cal D'}[\blocks(S,\pi)] \notag\\
	=& \Pr_{S \sim \cal D'}[\blocks(S,\pi) \land f_{G'}(S)= f_G(S)] \notag\\
	&+ \Pr_{S \sim \cal D'}[\blocks(S,\pi)\land f_{G'}(S) \ne f_G(S)] \notag\\
	\leq& \Pr_{S \sim \cal D'}[\blocks(S,\pi) \land f_{G'}(S)= f_G(S)]+ \frac{\eps}{2} \label{eq:SblocksFDiff}\\
	=& \Pr_{S \sim \cal D''}[\blocks(S,\pi)] + \frac{\eps}{2} \leq  \frac{\eps}{2} + \frac{\eps}{2} = \eps\label{eq:SblocksD''}
	\end{align}
	By construction of $G'$ and $\pi$, lines \eqref{eq:SblocksFDiff} and \eqref{eq:SblocksD''} hold with confidence $1- \frac{\delta}{2}$ each. We conclude that $\Pr_{S \sim \cal D}[S \text{ strongly blocks } \pi] \leq \eps$
	with confidence at least $1 - \delta$. Since the constructions of $G'$ and $\pi$ both require a polynomial number of samples from $\cal D$, $\calH^*$ is efficiently PAC stabilizable.
\end{proof}

Theorem~\ref{thm:conitzergraphs}, while interesting in its own right, provides us with only a partial understanding of the problem: if all one is given is positive examples, it is possible to find a tree structure that is consistent with all connected coalitions. In what follows, we study a more general question of whether we can find a forest consistent with both positive (connected coalitions) and negative (disconnected coalitions) examples. As we show in Theorem~\ref{thm:tree:hardness:consistency}, introducing the possibility of negative examples makes the problem computationally intractable, even if we restrict ourselves to the hypothesis class of paths. Hence, forests cannot be PAC learned efficiently. It is interesting to note that Theorem~\ref{thm:hg-stabilizable} could be achieved despite this negative result.

\subsection{The Complexity of Constructing Consistent Trees}\label{sec:consistentpath}
We now argue that deciding whether there exists a forest consistent with both positive and negative examples is computationally intractable; in fact, this claim holds even when the desired forest is a path. This result stands in sharp contrast to known computational results in the literature; indeed, there are several efficient algorithms for such restricted networks when only connected coalitions are taken into account \footnote{The problem of deciding the existence of a path consistent with connected coalitions is equivalent to the problem of determining whether the intersection graph of a hypergraph is an interval, which is also closely related to testing the consecutive ones property of a matrix (see, e.g. the survey by \citet{dom2009consecutive} for more details).} \cite{booth1976,korte1987,Corneil1998,fulkerson1965,Habib2000,Kratsch2006,Hsu1999}. 

Specifically, we are given $m$ samples of node subsets $S_1,\dots,S_m$; each subset $S_j$ is labeled by a function $\ell_G$ such that 
\begin{align}
\ell_G(S_j) = \begin{cases}1 & \mbox{if }S_j \mbox{ is connected in }G\\ 0 & \mbox{otherwise.}\end{cases}\label{eq:tree-connected}
\end{align}
We say that a graph $G^* = \tup{V,E^*}$ is {\em consistent with} $G=\tup{V,E}$ over the samples $\calS = \{S_1,\ldots,S_m\}\subseteq V$ if and only if $\ell_{G^*}(S_j) = \ell_{G}(S_j)$ for all $j \in [m]$.
Our objective is to find a forest $T^*$ such that $\ell_{T^*}(S_j) = \ell_G(S_j)$ for all $j$. Theorem~\ref{thm:tree:hardness:consistency} states that it is NP-hard to determine whether such a graph exists. 

\begin{theorem}\label{thm:tree:hardness:consistency}
	Given a family of subsets $\calS \subseteq 2^N$ such that each set in $\calS$ is of size at most $3$, and a mapping $\ell:\calS \rightarrow \{1,0\}$, it is NP-hard to decide whether there exists a path $T^*=\tup{N,E}$ such that $\ell_{T^*}(S) = \ell(S)$ for each $S \in \calS$. The result also holds when $T^*$ is a forest.
\end{theorem}
\begin{proof}
	We will first show that it is NP-hard to decide whether there exists a path consistent with both positive and negative samples; later we will show that how the reduction can be extended to forests. 
	
	Our reduction is from a restricted version of \SAT. Specifically, we consider \TBSAT. Recall that in this version of \SAT, each clause contains at most $3$ literals, and each variable occurs exactly twice positively and twice negatively; this problem is known to be NP-complete \cite{Berman2003}. 
	
	\noindent{\em Idea}: consider a formula $\phi$ with a variable set $X=\{x_1,x_2,\ldots,x_n\}$ and clause set $C=\{c_1,c_2,\ldots,c_m\}$, 
	where for each variable $x_i \in X$ we write $x_i(1)$ and $x_i(2)$ for the two positive occurrences of $x_i$, 
	and $\bar x_i(1)$ and $\bar x_i(2)$ for the two negative occurrences of $x$. We will have one {\em clause gadget} $C_j=\{c_{j}(1),c_{j}(2)\}$ for each clause $c_j \in C$ and one variable gadget $V_i=\{v_i(1),v_i(2)\}$ for each variable $x_i \in X$. Most player arrangements will be inconsistent with the pair $\tup{\calS,\ell}$ unless the following holds:
	\begin{itemize}
		\item For each clause $c_j \in C$, a literal player contained in a clause $c_j$ connects the players from a clause gadget.
		\item For each variable $x_i \in X$, either the pair of positive literal players $x_i(1), x_i(2)$ or the pair of negative literal players ${\bar x_i(1)},{\bar x_i(2)}$ connects the players from a variable gadget.
	\end{itemize}
	
	Hence, one can think of the variable gadgets as forcing a path to make a choice between setting $x_i$ {\em true} and setting $x_i$ {\em false}; each clause gadget ensures that the resulting assignment is satisfiable. 
	
	\noindent{\em Construction details}: 
	For each variable $x_i \in X$, we introduce two {\em variable players} $v_i(1)$ and $v_i(2)$, and four literal players 
	$$x_i(1), x_i(2), {\bar x_i(1)},{\bar x_i(2)},$$ which correspond to the four occurrences of $x$. 
	For each clause $c_j\in C$, we introduce two {\em clause players} $c_{j}(1)$ and $c_{j}(2)$. Let $k:= (4n-m-2n)+1=2n-m+1$. We introduce $k$ {\em garbage collectors} $g_1,g_2,\ldots,g_k$, and two {\em leaf} players $s$ and $t$. Intuitively, garbage collectors will be used to connect the literal players that do not appear in any clause or variable gadget. 
	
	Our set $\calS$ of samples consists of three subfamilies $\cal C$, $\calU_2$, and $\calU_3$: the sets in $\cal C$ correspond to the connectivity constraints, the sets in $\calU_2$ correspond to disconnected coalitions of size $2$, and the sets in $\calU_3$ correspond to disconnected coalitions of size $3$. 
	
	First, we construct the set $\calC$ that constitutes of 
	\begin{itemize}
		\item the four pairs $\{s,c_{1}(1)\}$, $\{c_{m}(2),v_{1}(1)\}$, $\{v_{n}(2),g_1\}$, $\{g_k,t\}$;
		\item the consecutive pairs $\{c_{j}(2),c_{j+1}(1)\}$ for $j \in [m-1]$; and 
		\item the consecutive pairs $\{v_i(2),v_{i+1}(1)\}$ for $i \in [n-1]$.
	\end{itemize}
	
	We next construct the negative samples $\calU_2$ and $\calU_3$ as follows. The family $\calU_2$ is the set of all player pairs, except for the following:
	\begin{itemize}
		\item the pairs in $\calC$.
		\item the pairs of a variable player and its corresponding literal player, i.e., the pairs of the form $\{v_{i}(h),x_{i}(h)\}$ or $\{v_{i}(h),{\bar x_{i}(h)}\}$.
		\item the pairs of a clause player and a literal player contained in it, i.e., the pairs of the form $\{c_{j}(h),y\}$ where $y$ is a literal player in a clause $c_j$.
		\item the pairs of positive literal players or negative literal players of each variable, i.e., the pairs of the form $\{x_i(1),x_i(2)\}$ or $\{{\bar x_i(1)},{\bar x_i(2)}\}$.
		\item the pairs of a literal player and a garbage collector, i.e., the pairs of the form $\{x_{i}(h),g_{k'}\}$ or $\{{\bar x_{i}(h)},g_{k'}\}$.
	\end{itemize}
	
	In a path consistent with the samples, each variable player can share an edge with its literal player; and each clause player can share an edge with a literal player contained in it. 
	
	The family $\calU_3$ consists of the following player triples:
	\begin{itemize}
		\item triples of the form $\{v_i(1),x_i(1),y\}$ where $y \neq v_{i-1}(2)$ and $y \neq x_i(2)$, and the triples of the form $\{v_i(1),{\bar x_i(1)},y\}$ where $y \neq v_{i-1}(2)$ and $y \neq {\bar x_i(2)}$; and
		\item the triples of the form $\{x_i(1), x_i(2),y\}$ where $y \neq v_i(1)$ and $y \neq v_i(2)$, and the triples of the form $\{{\bar x_i(1)},{\bar x_i(2)},y\}$ where $y \neq v_i(1)$ and $y \neq v_i(2)$.
	\end{itemize}
	Here $v_{0}(2)=c_{m}(2)$ and $c_{0}(2)=s$. 	
	The above constraints mean that if a variable player $v_i(1)$ and its positive literal player $x_i(1)$ (respectively, its negative literal player ${\bar x_i(1)}$) are adjacent, then the player $x_i(1)$ can be only adjacent to the other positive literal player $x_i(2)$ (respectively, the other negative literal player ${\bar x_i(2)}$), which can then be only adjacent to the other variable player $v_i(2)$. 
	
	Finally, for each $S \in \calS$ we set $\ell(S)=1$ if and only if $S \in \calC$.
	Note that the number of players in the instance is bounded by $O(n+m)$ and the number of sets in $\calS$ is bounded by $O(n^2+m^2)$.
	
	\noindent{\em Correctness}:
	We will now show that $\phi$ is satisfiable if and only if there exists a path consistent with $\tup{\calS,\ell}$. 
	
	\noindent
	$\Longrightarrow$: Suppose that there exists a truth assignment $f: X \rightarrow \{\mbox{\em true},\mbox{\em false}\}$ that satisfies~$\phi$. First, since $f$ is a satisfiable assignment for $\phi$, for each clause gadget $C_j=\{c_{j}(1),c_{j}(2)\}$, we can select exactly one literal that satisfies a clause $c_j$; we connect the literal player with each of the clause players $c_{j}(1)$ and $c_{j}(2)$ by an edge. We combine all the clause gadgets by constructing an edge $\{c_{j}(2),c_{j+1}(1)\}$ for each $j \in [m-1]$. Now, we consider an assignment that gives the opposite values to $f$, and connect each variable gadget using the literals corresponding to this assignment. Specifically, for each variable gadget $V_i=\{v_i(1),v_i(2)\}$, if $x_i$ is set to {\em false}, we select its positive literal players and construct a path that consists of three edges $\{v_i(1),x_i(1)\}$, $\{x_i(1),x_i(2)\}$, and $\{x_i(2),v_i(2)\}$; similarly, if $x_i$ that is set to {\em true}, we select its negative literal players and construct a path that consists of three edges $\{v_i(1),{\bar x_i(1)}\}$, $\{{\bar x_i(1)},{\bar x_i(2)}\}$, and $\{{\bar x_i(2)},v_i(2)\}$. We then create an edge $\{v_i(2),v_{i+1}(1)\}$ for each $i \in [n]$, and merge the variable gadgets all together. 
	
	Finally, we construct a path over the rest of players, by aligning the garbage collectors $g_1,g_2,\ldots,g_k$ in increasing order of their index, and putting one of the remaining $k-1(=4n-m-2n)$ literal players into each consecutive pair of garbage collectors arbitrarily. 
	We then merge all the paths by creating the four edges $\{s,c_{1}(1)\}$, $\{c_{m}(2),v_{1}(1)\}$, $\{v_{n}(2),g_1\}$, and $\{g_k,t\}$;
	see Figure \ref{fig1} for an illustration. It is easy to verify that the resulting graph is a path consistent with the samples. 
	
	\noindent
	$\Longleftarrow$: Conversely, suppose that there is a path $T^*=\tup{N,E}$ consistent with $\tup{\calS,\ell}$, i.e., for each $S \in \calS$, $S$ is connected in $T^*$ if and only if $S \in \calC$. Since every pair in $\calC$ should be connected, the four pairs $\{s,c_{1}(1)\}$, $\{c_{m}(2),v_{1}(1)\}$,$\{v_{n}(2),g_1\}$, $\{g_k,t\}$ must form an edge in $T^*$. Similarly, we have $\{c_{j}(2),c_{j+1}(1)\} \in E$ for each $j \in [m-1]$; also, $\{v_i(2),v_{i+1}(1)\} \in E$ for each $i \in [n-1]$. Observe that both players $s$ and $t$ must be the leaves of the constructed path since these players are only allowed to have one neighbor; thus, every other player has degree $2$. Combining these observations, the definition of $\calU_3$ ensures that our path specifies a truth assignment for $X$. 
	
	\begin{lemma}
		For each $i \in [n]$, we have either  
		\begin{itemize}
			\item $\{v_i(1),x_i(1)\}, \{x_i(1),x_i(2)\}, \{x_i(2), v_i(2)\} \in E$; or
			\item $\{v_i(1),{\bar x_i(1)}\}, \{{\bar x_i(1)},{\bar x_i(2)}\}, \{{\bar x_i(2)}, v_i(2)\} \in E$. 
		\end{itemize}
	\end{lemma}
	\begin{proof}
		Take any $i \in [n]$. Since each variable player $v_i(1)$ is adjacent to $v_{i-1}(2)$, the other players who can be adjacent to $v_i(1)$ are its literal players $x_i(1)$ and ${\bar x_i(1)}$ due to the constrains in $\calU_2$. First, if players $v_i(1)$ and $x_i(1)$ are adjacent, the player $x_i(1)$ can be only adjacent to $x_i(2)$ since $v_{i-1}(2)$ is already adjacent to $v_i(1)$, which then implies that $x_i(2)$ can be only adjacent to $v_i(2)$ due to the constrains in $\calU_3$. Similarly, if $v_i(1)$ and ${\bar x_i(1)}$ are adjacent, ${\bar x_i(1)}$ can be only adjacent to ${\bar x_i(2)}$, which then can be only adjacent to $v_i(2)$. This completes the proof. 

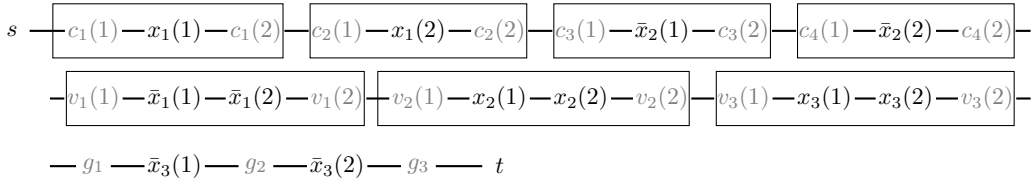
\begin{figure*}[tb]
		\centering
		\begin{tikzpicture}[scale=0.9, transform shape, every node/.style={minimum size=5mm, inner sep=1pt}]
		\node(0) at (0,0) {$s$};
		
		\draw (0.6,0.4) rectangle (4,-0.4);
		\node[gray](1) at (1.2,0) {$c_{1}(1)$};
		\node(2) at (2.4,0) {$x_1(1)$};
		\node[gray](3) at (3.6,0) {$c_1(2)$};
		
		\draw (4.4,0.4) rectangle (7.6,-0.4);
		\node[gray](4) at (4.8,0) {$c_{2}(1)$};
		\node(5) at (6,0) {$x_1(2)$};
		\node[gray](6) at (7.2,0) {$c_{2}(2)$};
		
		\draw (8,0.4) rectangle (11.2,-0.4);
		\node[gray](7) at (8.4,0) {$c_{3}(1)$};
		\node(8) at (9.6,0) {${\bar x_2(1)}$};
		\node[gray](9) at (10.8,0) {$c_{3}(2)$};
		
		\draw (11.6,0.4) rectangle (14.8,-0.4);
		\node[gray](10) at (12,0) {$c_{4}(1)$};
		\node(11) at (13.2,0) {${\bar x_2(2)}$};
		\node[gray](12) at (14.4,0) {$c_{4}(2)$};
		\node(13) at (15.3,0) {};
		\draw[-, >=latex,thick] (0)--(1) (2)--(1) (2)--(3) (4)--(3) (4)--(5) (6)--(5) (6)--(7) (8)--(7) (8)--(9) (9)--(10) (10)--(11) (11)--(12) (12)--(13);
		
		\node(14) at (0.3,-1) {};

		\draw (0.8,-0.6) rectangle (5.2,-1.4);
		\node[gray](15) at (1.2,-1) {$v_{1}(1)$};
		\node(16) at (2.4,-1) {${\bar x_1(1)}$};
		\node(17) at (3.6,-1) {${\bar x_1(2)}$};
		\node[gray](18) at (4.8,-1) {$v_{1}(2)$};
		
		\draw (5.4,-0.6) rectangle (10,-1.4);
		\node[gray](19) at (6,-1) {$v_{2}(1)$};
		\node(20) at (7.2,-1) {$x_2(1)$};
		\node(21) at (8.4,-1) {$x_2(2)$};
		\node[gray](22) at (9.6,-1) {$v_{2}(2)$};
		
		\draw (10.4,-0.6) rectangle (14.8,-1.4);
		\node[gray](23) at (10.8,-1) {$v_{3}(1)$};
		\node(24) at (12,-1) {$x_3(1)$};
		\node(25) at (13.2,-1) {$x_3(2)$};
		\node[gray](26) at (14.4,-1) {$v_{3}(2)$};
		
		\node(27) at (15.3,-1) {};
		
		\draw[-, >=latex,thick]  (14)--(15) (15)--(16) (16)--(17) (17)--(18) (18)--(19) (19)--(20) (20)--(21) (21)--(22) (23)--(22) (23)--(24) (24)--(25) (25)--(26) (26)--(27);
		
		\node(28) at (0.3,-2) {};
		\node[gray](29) at (1.2,-2) {$g_1$};
		\node(30) at (2.4,-2) {${\bar x_3(1)}$};
		\node[gray](31) at (3.6,-2) {$g_2$};
		\node(32) at (4.8,-2) {${\bar x_3(2)}$};
		\node[gray](33) at (6,-2) {$g_{3}$};
		\node(34) at (7.2,-2) {$t$};
		
		\draw[-, >=latex,thick] (28)--(29) (29)--(30) (30)--(31) (31)--(32) (32)--(33) (33)--(34) ;
		\end{tikzpicture}
		\caption{Graph constructed for the formula $\phi=(x_1 \lor x_2 \lor x_3) \land (x_1 \lor x_2 \lor {\bar x_3}) \land ({\bar x_1} \lor {\bar x_2} \lor x_3) \land ({\bar x_3} \lor {\bar x_2} \lor {\bar x_3})$ in the proof of Theorem \ref{thm:tree:hardness:consistency}. The formula is satisfied by the mapping $f$ that assigns the opposite value to the literals connected to each variable gadget $V_i=\{v_i(1),v_i(2)\}$, i.e., $f(x_1)=\mbox{\em true}$, $f(x_2)=\mbox{\em false}$, and $f(x_3)=\mbox{\em false}$. 
			\label{fig1}
		}
	\end{figure*}
	\end{proof}
Now take the truth assignment $f$ that sets the variable $x_i$ to {\em true} if and only if its negative literal players ${\bar x_i(1)}$ and ${\bar x_i(2)}$ are adjacent to variable players $v_i(1)$ and $v_i(2)$. This assignment can be easily seen to satisfy $\phi$. Indeed, for each clause $c_j \in C$, the clause player $c_{j}(2)$ must be adjacent to a literal player in $c_j$, since each clause player $c_{j}(2)$ is adjacent to $c_{j-1}(2)$ and the only other players who can be adjacent to $c_{j}(2)$ are their literal players contained in it; such a literal player corresponds to an occurrence appearing in the assignment $f$ and satisfies a clause $c_j$.
	
	To extend the above reduction to forests, given an instance of  \TBSAT, we create the same player set $N$ and family $\calS$ as above together with additional constraints; specifically, we indicate the two leaves $s$ and $t$ by forcing $(n-1)$-player coalitions $N \setminus \{s\}$ and $N \setminus \{t\}$ to be connected, and all other $(n-1)$-player coalitions to be disconnected. Thus, if there is a forest consistent with the samples, then it cannot have more than two leaves, namely, the graph must be a path.
\end{proof}

To conclude, if one allows observations of both connected and disconnected components, finding a forest consistent with samples is computationally intractable. We note that this does not preclude the existence of efficient heuristics computing consistent forest structures in practice: as previously mentioned, inferring PAC approximations of forest structures has a low communication complexity (Theorem~\ref{thm:finite-learn}); thus, given access to strong MILP solvers, we believe that identifying consistent forest structures should be easy in practice.

\section{Conclusions and Future Work}
This work establishes a strong connection between interaction structure and the ability to guarantee approximate stability in hedonic games; simply put, we show that if one only knows the underlying interaction structure and nothing more, then one can only obtain PAC stable outcomes if the underlying interaction structure is very well-behaved, i.e. a forest. This result seems to imply a natural tradeoff: our work assumes very little knowledge about underlying player preferences, and thus requires a lot of structure; \citet{sliwinski2017hedonic} make no assumptions on the underlying interaction network, but assume a more restricted player preference model. It would be interesting to explore `intermediate' cases; that is, suppose we make some structural assumptions on the interaction network, what classes of player preferences admit PAC stable outcomes? 

We make use of tools from computational learning theory in order to analyze hedonic coalition formation. We believe that as a research paradigm, this is a useful and important methodological approach. Hedonic games (and cooperative games in general) have, by and large, seen sparse application. Other game-theoretic methods have been successfully applied by taking a problem-oriented approach (e.g. stable matching for resident-hospital allocation \cite[Chapter 1.1]{kleinberg2006algorithm}, or Stackelberg games in the security domain~\cite{tambe2011security}); a concrete problem modeled and solved by a hedonic game framework has not yet been identified, to the best of our knowledge; this is despite the wealth of potential application domains, and rich data environments available nowadays (in particular, social network datasets would be particularly agreeable to the type of analysis presented in this work). Our work makes a fundamental connection between data, community structure, and game-theoretic solution concepts; a connection that we hope will result in a more applicable model of strategic collaborative behavior.

\bibliographystyle{aaai}
\bibliography{abbshort,learninghedonicrefs}

\end{document}